\documentclass[a4paper,UKenglish,cleveref,autoref,thm-restate]{lipics-v2021}

\nolinenumbers

\usepackage{microtype}
\usepackage{nccmath}
\usepackage[nocompress]{cite}
\usepackage{tikz}
\usepackage{circuitikz}
\usetikzlibrary{automata, positioning}
\usepackage{tabularx}
\usepackage{booktabs}
\usepackage{xspace}
\usepackage[noEnd=true,indLines=false]{algpseudocodex}

\usepackage[only,llbracket,rrbracket]{stmaryrd}
\usepackage{mathrsfs}
\usepackage{mathtools}
\usepackage{adjustbox}

\renewcommand{\geq}{\geqslant}
\renewcommand{\leq}{\leqslant}

\newcommand{\ie}{i.e.,~}

\newcommand{\sse}{\subseteq}

\newcommand{\ignore}[1]{}

\newcommand{\abs}[1] {\ensuremath\left|#1\right|}

\newcommand{\set}[2] {\ensuremath{\left\{#1 \mid #2\right\}}}

\newcommand{\N}{\mathbb{N}}

\newcommand{\gen}[1]{\langle #1 \rangle}

\newcommand{\Gen}[2]{\left< \mathinner{#1} \mid \mathinner{#2}\right>}

\newcommand{\NP}{\ensuremath{\mathsf{NP}}\xspace}

\renewcommand{\P}{\ensuremath{\mathsf{P}}\xspace}

\newcommand\nindent{.5pt}
\newcommand\noverline[1]{%
	\kern\nindent\overline{\kern-\nindent#1\kern-\nindent}\kern\nindent}

\newcommand{\cX}{\mathcal{X}}

\newcommand{\sD}{\mathscr{D}}

\newcommand{\sC}{\mathscr{C}}

\newcommand{\compproblem}[3][]{%
	\par\vspace{0.125cm plus 0.1cm minus 0.05cm}\begin{tabularx}{\textwidth-2\parindent}{lX}%
		\if\relax\detokenize{#1}\relax%
		\else%
		\textnormal{\textbf{Constant:}}&#1\\%
		\fi%
		\textnormal{\textbf{Input:}}&#2\\%
		\textnormal{\textbf{Question:}}&#3\\%
	\end{tabularx}\vspace{0.125cm plus 0.1cm minus 0.05cm}\par%
}

\newcommand{\comproblem}[3][]{%
	\par\vspace{0.125cm plus 0.1cm minus 0.05cm}\begin{tabularx}{\textwidth-2\parindent}{lX}%
		\if\relax\detokenize{#1}\relax%
		\else%
		\textnormal{\textbf{Constant:}}&#1\\%
		\fi%
		\textnormal{\textbf{Input:}}&#2\\%
		\textnormal{\textbf{Output:}}&#3\\%
	\end{tabularx}\vspace{0.125cm plus 0.1cm minus 0.05cm}\par%
}

\theoremstyle{plain}

\makeatletter % bold heading for remarks
\def\th@remark{%
	\thm@headfont{%
		\textcolor{lipicsGray}{$\blacktriangleright$}\nobreakspace\sffamily\bfseries}%
	\normalfont % body font
	\thm@preskip\topsep \divide\thm@preskip\tw@
	\thm@postskip\thm@preskip
}
\makeatother

\theoremstyle{plain}

\makeatletter
\providecommand\iitem{}
\providecommand\qitem{}
\newcommand\decproblem@iitem@label{\rlap{Input.}\phantom{Question.}}
\newcommand\decproblem@qitem@label{Question.}

\makeatother

\makeatletter
\providecommand\iitem{}
\providecommand\qitem{}
\newcommand\coproblem@iitem@label{\rlap{Input.}\phantom{Output.}}
\newcommand\coproblem@qitem@label{Output.}

\makeatother

\newcommand{\dEqn}[2][]{\textup{\textsc{eqn${}_{\mathbf{#1}}\expandafter\ifx\expandafter\relax\detokenize{#2}\relax\else(#2)\fi$}}}
\newcommand{\dEqnSys}[2][]{\textup{\textsc{eqn${}^\ast_{\mathbf{#1}}\expandafter\ifx\expandafter\relax\detokenize{#2}\relax\else(#2)\fi$}}}

\newcommand{\fA}{\mathfrak A}
\newcommand{\fB}{\mathfrak B}
\newcommand{\fC}{\mathfrak C}
\newcommand{\fI}{\mathfrak I}
\DeclareMathOperator{\PCSP}{PCSP}
\DeclareMathOperator{\CSP}{CSP}
\DeclareMathOperator{\Proj}{Proj}
\DeclareMathOperator{\Meta}{Meta}
\DeclareMathOperator{\PMeta}{PMeta}
\DeclareMathOperator{\PCreaMeta}{PCMeta}
\DeclareMathOperator{\Pol}{Pol}
\DeclareMathOperator{\Aut}{Aut}

\newcommand{\SigAb}{\Sigma_{\textnormal{abheap}}}
\newcommand{\SigHeap}{\Sigma_{\textnormal{heap}}}
\newcommand{\SigMaltsev}{\Sigma_{\textnormal{Maltsev}}}

\title{The complexity of finding coset-generating polymorphisms and the promise metaproblem}
\titlerunning{Coset-generating Polymorphisms}

\author{Manuel Bodirsky}{Institut f\"ur Algebra, TU Dresden}{manuel.bodirsky@tu-dresden.de}{https://orcid.org/0000-0001-8228-3611}{Received funding from the ERC (Grant Agreement no. 101071674, POCOCOP). Views and opinions expressed are however
those of the authors only and do not necessarily reflect those of the European Union or the European Research
Council Executive Agency.}
\author{Armin Weiß}{FMI, University of Stuttgart}{armin.weiss@fmi.uni-stuttgart.de}{https://orcid.org/0000-0002-7645-5867}{Supported by the German Research Foundation (DFG) grant WE 6835/1-2.} % \\ Universitätsstraße 38, 70569 Stuttgart, Germany
\authorrunning{M.~Bodirsky, A.~Weiß}

\hideLIPIcs

\Copyright{Manuel Bodirsky and Armin Weiß}


\begin{document}
	
	\maketitle
	
	\begin{abstract}
		We show that the metaproblem for coset-generating polymorphisms 
        is \NP-complete, answering a question of Chen and Larose: given a finite structure, the computational question is whether this structure
        has a polymorphism of the form $(x,y,z) \mapsto x y^{-1} z$ with respect to some group; such operations are also called \emph{coset-generating}, or \emph{heaps}.
        
		Furthermore, we introduce a promise version of the metaproblem, parametrised by two polymorphism conditions $\Sigma_1$ and $\Sigma_2$ and defined analogously to the promise constraint satisfaction problem.
        We give sufficient conditions under which the promise metaproblem for $(\Sigma_1,\Sigma_2)$ is in \P and under which it is \NP-hard.
        In particular, the promise metaproblem is in \P if $\Sigma_1$ states the existence of a Maltsev polymorphism and $\Sigma_2$ states the existence of an abelian heap polymorphism~-- despite the fact that neither the metaproblem for $\Sigma_1$ nor the metaproblem for $\Sigma_2$ is known to be in \P.
        We also show that the creation-metaproblem for Maltsev polymorphisms, under the promise that a heap polymorphism exists, is in \P if and only if there is a uniform polynomial-time algorithm for CSPs with a heap polymorphism. 
	\end{abstract}

%\tableofcontents

	\section{Introduction}
	After the resolution of the finite-domain constraint satisfaction problem (CSP) dichotomy conjecture~\cite{Zhuk20} (announced independently by Bulatov~\cite{BulatovFVConjecture} and by Zhuk~\cite{ZhukFVConjecture}), one of the central open problems in the area is whether there exists a \emph{uniform} polynomial-time algorithm for tractable CSPs, i.e., an algorithm that decides in polynomial time whether there exists a homomorphism from a given relational structure $\fA$ to a given relational structure $\fB$, under the assumption that $\CSP(\fB)$ is in \P. 
	Here the structures are given by lists of tuples, that is, one list of tuples for each of the relations.
    Neither Bulatov’s nor Zhuk’s algorithm runs uniformly in polynomial time, because the running time of these algorithms depends superpolynomially on the size of the input structure $\fB$.

	The existence of a uniform polynomial-time algorithm for tractable CSPs is already open under the much stronger assumption that the given structure $\fB$ has a \emph{Maltsev polymorphism}, i.e., a polymorphism $m \colon B^3 \to B$ satisfying $m(x,x,y) = m(y,x,x) = y$ for all $x,y \in B$. 
	An important source of such operations are groups as the operation $(x,y,z) \mapsto x y^{-1} z$ with respect to some group satisfies these identities; such operations have been called \emph{coset-generating polymorphisms} (terminology used in~\cite{MetaChenLarose}; for the motivation for this name, see \cref{lem:folklore}), which are also called \emph{heaps} (see \cref{sect:coset-gen}). 
	For structures $\fB$ with a Maltsev polymorphism, the polynomial-time tractability of $\CSP(\fB)$ has been shown prior to \cite{BulatovFVConjecture,ZhukFVConjecture} and with a simpler algorithm (Bulatov and Dalmau~\cite{Maltsev}). But again, this algorithm is \emph{not} a uniform polynomial-time algorithm. 
    Nevertheless, the algorithm of Bulatov and Dalmau is \emph{semiuniform} in the sense that, if the algorithm is not only given $\fB$ but also a Maltsev polymorphism of $\fB$, then it can answer the question whether there exists a homomorphism from $\fA$ to $\fB$ in polynomial time.\footnote{Also the polynomial-time algorithms of Zhuk and of Bulatov are semiuniform; instead of a Maltsev polymorphism, they work if a so-called  \emph{weak near unanimity polymorphism} of $\fB$ is part of the input.}
	Chen and Larose~\cite{MetaChenLarose} pointed out that a uniform polynomial-time algorithm for constraints with a Maltsev polymorphism would thus follow from the existence of a polynomial-time algorithm that takes a structure $\fB$ as input and computes a Maltsev polymorphism of $\fB$; this problem is called the \emph{creation-metaproblem} (or \emph{creation-metaquestion}, see~\cite{MetaChenLarose}) for Maltsev polymorphisms.
    Conversely, if there exists a uniform polynomial-time algorithm for structures $\fB$ with a Maltsev polymorphism, then the creation-metaproblem for Maltsev polymorphisms is in \P (a consequence of Theorem 4.7 in~\cite{MetaChenLarose}; we use \P also for function-\P). 

    Some of the recently proposed candidates for a uniform polynomial-time algorithm for tractable CSPs were refuted using structures with a coset-generating polymorphism:
    \emph{Datalog reductions to systems of linear equations over ${\mathbb Z}$} (a method proposed in~\cite{DalmauOprsalLocal}) fail to solve the CSP for a structure with the coset-generating polymorphism for the symmetric group $S_{18}$~\cite{LichterP25}, and \emph{singleton BLP+AIP} fails to solve the CSP for a structure with the coset-generating polymorphism for the dihedral group $D_8$~\cite{DimaSingleton}.
	It is thus natural to determine the complexity of the creation-metaproblem for coset-generating polymorphisms (i.e., heaps) instead of Maltsev polymorphisms and the corresponding decision problem, which is simply called the \emph{metaproblem} (or \emph{metaquestion}). For the latter, the task is to decide whether the polymorphisms of a given finite structure $\fB$ (where the relations are given explicitly by lists of tuples) satisfy a certain polymorphism condition such as having a Maltsev polymorphism, or having a coset-generating polymorphism.
	And indeed, Chen and Larose~\cite[Section 8]{MetaChenLarose} ask: ``Can anything be said about the complexity of deciding the presence of a coset-generating polymorphism?''.
	In this article we prove (in \cref{sect:coset-gen-proof}) that deciding whether a given structure has a coset-generating polymorphism is an \NP-complete problem, thus answering the question of Chen and Larose.

	\begin{theorem}\label{thm:groupCSPNPcompleteIntro}
		The following problem is \NP-complete:
		\compproblem{A finite structure $\fB$.}{Does $\fB$ have a coset-generating polymorphism?}
	\end{theorem}

    While this does not rule out the existence of a uniform polynomial-time algorithm for CSPs with coset-generating polymorphisms, it shows that such a possible algorithm cannot use the approach to first compute the coset-generating polymorphism and then use Bulatov and Dalmau's semiuniform algorithm~\cite{Maltsev}.

	One might ask whether the \NP-hardness result in \cref{thm:groupCSPNPcompleteIntro} still holds for coset-generating polymorphisms with respect to restricted classes of groups.
	Indeed, our proof only uses groups of order $4p$ for $p \geq 5$ prime~-- which, in particular, are all metabelian and, indeed, have an abelian normal subgroup of index 4.
	
	Coset-generating operations $f$ that come from \emph{abelian} groups are called \emph{abelian heaps}. 
	If $\fB$ has a polymorphism that is an abelian heap, then a uniform polynomial-time algorithm is known (based on \emph{affine integer programming} (AIP) introduced in~\cite{BrakensiekG20}).
	We show that this algorithm can be used to compute in polynomial time for a given finite structure $\fB$ a Maltsev polymorphism when the existence of an abelian heap polymorphism is promised.
	
	More generally, we introduce the \emph{promise-metaproblem}, which is defined analogously to the \emph{promise CSP} (for literature on this very active area, see, e.g.,~\cite{BartoBKO21,BrakensiekGWZ20}): if $\Sigma_1$ and $\Sigma_2$ are polymorphism conditions (formal definitions can be found in \cref{sect:prelims}) such that $\Sigma_1$ implies  $\Sigma_2$, then $\PMeta(\Sigma_1,\Sigma_2)$ is the problem of deciding for a given finite structure $\fB$ whether $\Pol(\fB)$ satisfies $\Sigma_1$, or whether it does not even satisfy $\Sigma_2$. 
	For instance, if $\Sigma_1$ is the condition to have an abelian heap polymorphism, and 
	$\Sigma_2$ is the condition to have a Maltsev polymorphism, 
	then the argument above shows that $\PMeta(\Sigma_1,\Sigma_2)$ is in \P despite the fact that neither of the metaproblems for $\Sigma_1$ and for $\Sigma_2$ is known to be in \P. 
	More generally, we strengthen a result of Chen and Larose~\cite[Corollary 4.9]{MetaChenLarose} about the tractability of certain metaproblems    
    and prove that $\PMeta(\Sigma_1,\Sigma_2)$ is in \P if 
	\begin{itemize}
		\item $\Sigma_1$ is idempotent and there exists a uniform polynomial-time algorithm for $\CSP(\fB)$ for structures $\fB$ with $\Pol(\fB) \models \Sigma_1$, and 
		\item $\Sigma_2$ is a \emph{linear strong Maltsev condition} (\cref{prop:inP}).\footnote{A \emph{strong Maltsev condition} is simply a finite set of height one conditions, see Section~\ref{sect:prelims} for definitions.  This is standard terminology, but we will avoid using it to prevent confusing with the Maltsev identities mentioned in the introduction. 
    Linear strong Maltsev conditions are also called \emph{minor conditions} in the Promise CSP literature.}
	\end{itemize}
    Moreover, under the same conditions and provided that there exists a semiuniform polynomial-time algorithm for CSPs with polymorphisms satisfying $\Sigma_2$, the uniform CSP for $\Sigma_1$ is in \P if and only if the creation variant of the promise metaproblem is in \P.\footnote{Promise metaproblem should not be confused with meta-problems for promise CSPs, as studied e.g.~in~\cite{Larrauri25}.}
    In particular, we conclude that there is a uniform polynomial-time algorithm for CSPs with a coset-generating polymorphism if and only if the promise creation-metaproblem for coset-generating polymorphisms and Maltsev polymorphisms is in \P.

	Also one of the main hardness results by Chen and Larose about the metaproblem can be generalized to the larger class of promise metaproblems:
	we observe that if $\Sigma_1$ and $\Sigma_2$ are \emph{non-trivial} and \emph{consistent height one} strong Maltsev conditions 
    (see  Section~\ref{sect:polymorphisms} for definitions)  
    such that $\Sigma_1$ implies $\Sigma_2$, then $\PMeta(\Sigma_1,\Sigma_2)$ is \NP-complete (\cref{prop:negative}).
  
	For a formal discussion of the complexity of uniform algorithms and metaproblems, we need to be more specific about how relational structures are given to a computer and how to measure the size of these representations. This can be found in \cref{sect:prelims}. 
    In Section~\ref{sect:meta} we introduce the promise metaproblem and promise creation-metaproblem and observe that a hardness result of Chen and Larose has a natural extension to the promise setting (Proposition~\ref{prop:negative}).
    In Section~\ref{sect:uniform} we discuss the connection between the promise metaproblem and uniform algorithms for the CSP, and prove a generalization of a known tractability result for metaproblems to promise metaproblems (Proposition \ref{prop:inP}).
	In \cref{sect:coset-gen} we give some background on coset-generating polymorphisms, and in \cref{sect:coset-gen-proof} we prove our main result about the computational complexity of the metaproblem for coset-generating polymorphisms. We conclude and discuss open problems for future research in Section~\ref{sect:open}.

	\section{Structures, Homomorphisms, Polymorphisms} 
	\label{sect:prelims}
	Let $k \in {\mathbb N}$ and $\tau \in {\mathbb N}^k$. 
	A $\tau$-structure $\fA$ consists of a set $A$ and for every $i \in \{1,\dots,k\}$ a relation $R^{\fA}_i \subseteq A^{\tau_i}$. 
	We refer to $\tau$ as the \emph{signature} of $\fA$ (of size $k$).
	All $\tau$-structures considered here will be \emph{finite}, i.e., $A$ is finite. 
	If a $\tau$-structure is considered as the input or output of some algorithm, we assume that its relations are represented by lists of tuples, i.e., the representation size of $\fA$ is $|A|+k+\sum_{i \in \{1,\dots,k\}} (\tau_i \cdot |R_i^{\fA}|)$. 
	A \emph{structure} is a $\tau$-structure for some signature $\tau$; if we later write that a computational problem takes a structure as an input, we mean that $\tau$ and $k$ are part of the input.
	
	A \emph{homomorphism} between two structures $\fA$ and $\fB$ with the same signature $\tau$ of size $k$ is a function $h \colon A \to B$ such that for every $i \in \{1,\dots,k\}$ we have 
	$(h(a_1),\dots,h(a_{\tau_i})) \in R^{\fB}_i$ whenever
	$(a_1,\dots,a_{\tau_i}) \in R^{\fA}_i$.
	If $\fA$ is a $\tau$-structure and $m \in {\mathbb N}$, then
	$\fB := \fA^m$ is the $\tau$-structure with domain $A^m$ 
	and \[R^{\fB}_i := \big \{\big((a^1_1,\dots,a^m_1),\dots,(a^1_{\tau_i},\dots,a^m_{\tau_i})\big) \mid (a^j_1,\dots,a^j_{\tau_i}) \in R_i^{\fA} \text{ for all } j \in \{1,\dots, m\} \big \}.\]
	A \emph{polymorphism} of a structure $\fB$ is a homomorphism from $\fB^m$ to $\fB$, for some $m \in {\mathbb N}$.
	In other words, a polymorphism is an operation $B^m \to B$ for some $m$ that, applied componentwise, preserves the relations of $\fB$. 
	Note that the set of all polymorphisms of $\fB$ is a \emph{clone}, i.e., it contains the projections and is closed under composition. 
	A finite structure $\fB$ is called a \emph{core} if every endomorphism\footnote{An endomorphism of $\fB$ is a homomorphism from $\fB$ to $\fB$, i.e., a unary polymorphism.} of $\fB$ is	an automorphism of $\fB$.

	The problem of deciding whether there exists a homomorphism
	from a $\tau$-structure $\fA$ to a $\tau$-structure $\fB$ is \NP-complete. 
	This problem is sometimes called the \emph{uniform CSP} (with both $\fA$ and $\fB$ being the input). 
	If $\fB$ is a fixed $\tau$-structure, then $\CSP(\fB)$ is the problem of deciding whether a given $\tau$-structure $\fA$ has a homomorphism to $\fB$; this is often referred to as the \emph{non-uniform CSP}. 
	The finite-domain CSP dichotomy states that $\CSP(\fB)$ is in \P if $\fB$ has a \emph{4-ary Siggers polymorphism},
	i.e., a polymorphism $s$ satisfying for all $a,r,e \in B$ 
	\[s(a,r,e,a)=s(r,a,r,e),\]
	and is \NP-complete otherwise. This specific identity was found by Siggers~\cite{Siggers}, building on results from~\cite{BartoKozikNiven}. The result proved there also implies that if a structure has a Maltsev polymorphism, then it also has a 4-ary Siggers polymorphism.
	Note that unlike some older articles, we do not require that Siggers operations are idempotent. We also mention that for finite structures $\fB$ the following are equivalent. 
	\begin{itemize}
		\item $\fB$ has a Siggers polymorphism;
		\item $\fB$ has a \emph{cyclic polymorphism}~\cite{Cyclic}, i.e., a polymorphism $c \colon B^k \to B$ for some $k \geq 2$ satisfying for all $x_1,\dots,x_k \in B$
		\[ c(x_1,\dots,x_{k}) = c(x_2,\dots,x_{k},x_1).\]
		\item $\fB$ has a \emph{weak near-unanimity polymorphism}~\cite{MarotiMcKenzie}, i.e., a polymorphism $w \colon B^k \to B$ for some $k \geq 2$ satisfying for all $x,y \in B$
		\[ w(y,x,\dots,x) = w(x,y,x,\dots,x) = \cdots = w(x,\dots,x,y).\]
	\end{itemize}
    For every fixed pair of relational $\tau$-structures $\fB$ and $\fC$ such that $\fB$ has a homomorphism to $\fC$,
	the \emph{Promise CSP for $(\fB,\fC)$}, denoted by $\PCSP(\fB,\fC)$, is defined as follows.
	
	\comproblem{A finite $\tau$-structure $\fA$.}
	{`Yes' if $\fA$ has a homomorphism to $\fB$,
		`No' if $\fA$ has no homomorphism to $\fC$.} 
	
	If $\fA$ has a homomorphism to $\fC$ but no homomorphism to $\fB$, then the algorithm solving $\PCSP(\fB,\fC)$ might answer arbitrarily. 
	Despite a powerful theory that is available for studying the complexity of promise CSPs~\cite{BartoBKO21,BrakensiekGWZ20,BrakGuruAlgebraic}, the complexity of $\PCSP(\fB,\fC)$ has not been classified for all pairs $(\fB,\fC)$. 
	For instance,
    if $K_{\ell}$ is the clique with $\ell$ vertices, then the complexity of 
    $\PCSP(K_n,K_m)$ has already been studied in~\cite{GareyJohnssonGraphCol}, and is conjectured to be \NP-hard for all $3 \leq \;  n<m$, but this remains open. There are problems of the form $\PCSP(\fA,\fB)$ in \P where $\CSP(\fA)$ and $\CSP(\fB)$ are \NP-hard~\cite{BrakensiekGuruswami19}. 
	
    \section{Polymorphism Conditions}\label{sect:polymorphisms}

	Let $\sigma$ be a set of function symbols.
	A \emph{$\sigma$-identity} is an expression of the form $s \approx t$ where $s$ and $t$ are $\sigma$-terms over some common set of variables $x_1,\dots,x_n$. 
    Such an identity is called 
    \begin{itemize}
    		\item \emph{linear} if each of the two terms contains at most one function symbol from $\sigma$;
		\item \emph{height one} if exactly one symbol from $\sigma$ occurs in each of the terms $s$ and $t$. 
    \end{itemize}   
    If $\Sigma$ is a (not necessarily finite) set of $\sigma$-identities 
	and $\sC$ is a set of operations over $B$, 
	then we write $\sC \models \Sigma$ if 
	we can interpret the function symbols from $\sigma$
	by operations from $\sC$ such that 
	for every identity $s \approx t$ from $\Sigma$ 
	and every instantiation of the variables in $s$ and $t$ by elements $b_1,\dots,b_n$ from $B$ we have that $s(b_1,\dots,b_n) = t(b_1,\dots,b_n)$.
	Otherwise, we write $\sC \not \models \Sigma$. 
	If $\Sigma = \{s \approx t\}$ consists of a single identity, we also write $\mathscr C \models s \approx t$ instead of $\mathscr C \models \{s \approx t\}$. 
	If $\Sigma_1$ and $\Sigma_2$ are sets of identities, we write $\Sigma_1 \models \Sigma_2$ if 
	for every clone $\sC$ with a finite domain we have that 
	$\sC \models \Sigma_2$ whenever 
	$\sC \models \Sigma_1$. 
	A set of $\sigma$-identities $\Sigma$ is called 
	\begin{itemize}
		\item \emph{idempotent} if $\Sigma$ contains for every $f \in \sigma$ the identity $f(x,\dots,x) \approx x$;
		\item \emph{trivial} if $\Proj \models \Sigma$, where $\Proj$ is the set of all projections over the set $\{0,1\}$ (and \emph{non-trivial} otherwise);
		\item \emph{consistent} if for every non-empty finite set $D$ there exist idempotent operations on $D$ that satisfy $\Sigma$ (we use the definition  from~\cite{MetaChenLarose}; also see Lemma~\ref{lem:consistent}).
	\end{itemize}

    We say a function $g \colon D^k \to D$ is an \emph{extension} of $f \colon C^k\to C$ for $C \sse D$ if $g_{|C^k} = f$ where $g_{|C^k}$ denotes the restriction of $g$ to $C^k$.

    \begin{lemma}\label{lem:extend_domain}
		Let $\Sigma$ be a set of linear $\sigma$-identities such that $\Sigma \not \models x \approx y$ and let $\mathscr C$ be a clone on some domain $C$ such that $\sC \models \Sigma$.
		Moreover, let $D = C \cup \{a\}$ for some $a \notin C$ and let $\sD$ be the set of all operations that extend some $f \in \sC$ (\ie\ $\sD = \{\hat f \colon D^k \to D \mid k\in \N, \hat f_{|C^k} \in \sC\}$).
        Then $\sD$ is a clone and $\sD \models \Sigma$.
	\end{lemma}

    \begin{proof} 
        First, let us observe that $\sD$ is a clone: Clearly, all projections are in $\sD$. Moreover, if $\hat f, \hat g_1, \dots, \hat g_k \in \sD$, then we have $ f,  g_1, \dots,  g_k \in \sC$ where $f, g_1, \dots, g_k$ denote the respective restrictions to $C^\ell$ for suitable $\ell$. Hence, because $\sC$ is a clone, also their composition is in $\sC$ (which is the same as the restriction of $\hat f( \hat g_1, \dots, \hat g_k) $); hence, also $\hat f( \hat g_1, \dots, \hat g_k) \in \sD$ and $\sD$ is a clone.

		To show that $\sD \models \Sigma$, we define an equivalence relation $\equiv$ on the set of linear $\sigma$-terms over a countable set of  variables $\cX$.
        Note that a linear $\sigma$-term is either a single variable or a term of the form $f(x_1, \dots, x_n)$ where $f \in \sigma$ is an $n$-ary function symbol and $x_1, \dots, x_n \in \cX$ (possibly with repetitions).
        We can also denote such a linear term in the form $s(y_1, \dots, y_k) = f(y_{i_1}, \dots, y_{i_n})$ where $y_1, \dots, y_k$ are pairwise distinct variables and $i_1, \dots, i_n \in \{1, \dots, k\}$.

        We define $\equiv$ to be the smallest equivalence relation which contains $\Sigma$ (recall that formally $\Sigma$ is a set of pairs of linear $\sigma$-terms) and which is closed under substituting variables for other variables.
        Thus, in particular, we have $\Sigma \models \{s \approx t \mid s \equiv t\}$.
		Observe that 
        \begin{enumerate}[(1)]
            \item\label{point1} as $\Sigma \not \models x \approx y$, we cannot have both $ f(x_1, \dots, x_k) \equiv x_i$ and $ f (x_1, \dots, x_k) \equiv x_j$ if $x_i \neq x_j$.
        \end{enumerate}
	  
        We may identify the function symbols in $\sigma$ with operations in $\sC$ that witness
        that $\sC$ satisfies  $\Sigma$.
        Using this identification, we further observe that,
        \begin{enumerate}[(1)]
            \setcounter{enumi}{1}
            \item\label{point2}  if $s(y_1, \dots, y_k) $ is a linear term 
            where $y_1, \dots, y_k \in \cX$ are distinct variables and if $s (y_1, \dots, y_k) \equiv y_i$ for some $i \in \{1,\dots,k\}$, then $s(c_1, \dots, c_k) = c_i$ for all $c_1, \dots, c_k \in C$.
        \end{enumerate}

		Now, for each $f\in \sigma$ identified with $f \colon C^k \to C$, let us describe an extension $\hat f \colon D^k \to D$ such that the resulting set of functions satisfies $\Sigma$.
        In order to do so, fix a map $\pi \colon D \to C$ which is the identity on $C$. 
        For $f\in \sigma$ and $a_1,\dots,a_k \in D$ we define 
        \[
            \hat f(a_1, \dots, a_k) :=
            \begin{cases}
                a_i & \parbox[t]{0.5\linewidth}{%
                        if there exists an injective $\xi \colon \{a_1, \dots, a_k\} \to \cX$\newline
                        such that
                        $f(\xi(a_1), \dots, \xi(a_k)) \equiv \xi(a_i)$
                    } \\
                f(\pi(a_1), \dots, \pi(a_k)) & \text{otherwise.} \vphantom{k^{k^{k^{k^k}}}}
            \end{cases}
        \]
        
 		By \eqref{point1}, it follows that $\hat f$ is well-defined (note that the definition depends on $\pi$, but $\pi$ is fixed~-- however, it does not depend on $\xi$, because the relation $\equiv$ is invariant under substitution of variables).
		Moreover, if $a_1, \dots, a_k \in C$, we have $\hat f(a_1, \dots, a_k) = f(a_1, \dots, a_k)$ (in the case that there exists an injective $\xi \colon \{a_1,\dots,a_k\} \to \cX$ with $f(\xi(a_1),\dots,\xi(a_k)) \equiv \xi(a_i)$,
        we rely on \eqref{point2}).
        
		Finally, the collection of all the $\hat f$ for $f\in \sigma$ satisfies $\Sigma$.	
        For identities involving  $f (x_1, \dots, x_k)$ with $f (x_1, \dots, x_k) \equiv x_i$, this is because we defined $\hat f(a_1, \dots, a_k) = a_i$; for other identities, this follows immediately,  because the set of all $f$'s (with $f\in \sigma$) already satisfies $\Sigma$.
	\end{proof}

    The purpose of the following lemma is to explain the concept of `consistency' for sets of linear identities; the second statement illustrates why we cannot drop the idempotence assumption in the first statement.
	  \begin{lemma}\label{lem:consistent}
		An idempotent set of linear identities $\Sigma$ is consistent if and only if $\Sigma \not \models x \approx y$.
		A set of height one identities $\Sigma$ is consistent if and only if $\Sigma \not \models f(x) \approx g(y)$. 
	\end{lemma}

	\begin{proof} 
 		Clearly, if $\Sigma \models x \approx y$, then all finite models of $\Sigma$ have a domain of size one. Hence, if $D$ has cardinality $n \geq 2$, then there is no set of operations on $D$ that satisfies $\Sigma$, so $\Sigma$ is not consistent. 
		On the other hand, if $\Sigma \not \models x \approx y$, by \cref{lem:extend_domain}, starting with a one-element domain, we find clones with arbitrary domain size that satisfy $\Sigma$. 
        If $\Sigma$ is idempotent, then we even find idempotent operations that satisfy $\Sigma$, so $\Sigma$ is consistent.
		
	    To see the second statement, observe that if $\Sigma \models f(x) \approx g(y)$ 
		and $\sC$ is a clone such that $\sC \models \Sigma$, then $\sC$ must contain a constant unary operation, which can only be idempotent if the domain $D$ of $\sC$ has at most one element. So $\Sigma$ is not consistent. Conversely, 
		suppose that 
		$\Sigma \not \models f(x) \approx g(y)$.
		It is known that in this case 
		${\mathscr D} \models \Sigma \cup \Delta$, where 
		${\mathscr D}$ is the clone of all idempotent operations on a two-element set $D$ (Proposition 3.2.5 in~\cite{VucajThesis}), and $\Delta$ are equations that express 
        that all functions mentioned in $\Sigma$ are idempotent.
        Let $S$ be a finite superset of $D$; then applying 
        Lemma~\ref{lem:extend_domain} 
        sufficiently many times we obtain 
        a clone on $S$ which also satisfies $\Sigma \cup \Delta$. Hence, $\Sigma$ is consistent. 
	\end{proof}
    
	Note that the existence of a Siggers polymorphism can be expressed by the height one identity 
	\[f(a,r,e,a) \approx f(r,a,r,e)\]
	and the existence of a Maltsev polymorphism can be expressed by the  idempotent two-element set of linear identities
	\begin{align}
        \SigMaltsev :=\{ m(x,x,y) \approx y, m(y,x,x) \approx y\}. 
		\label{eq:malt}
	\end{align}
	Both of these sets of identities are non-trivial and consistent.

\section{The Promise Metaproblem}\label{sect:meta}

	If $\Sigma$ is a set of $\sigma$-identities, 
	then the \emph{metaproblem for $\Sigma$}, denoted by $\Meta(\Sigma)$, is the following computational problem. 
	
	\comproblem{A finite structure $\fB$.}{`Yes' if $\Pol(\fB) \models \Sigma$, `No' otherwise.}
	
	The following was shown by Chen and Larose~\cite{MetaChenLarose}.
	
	\begin{theorem}[Theorem 6.2 in~\cite{MetaChenLarose}]
		\label{thm:h1} 
		Let $\Sigma$ be a non-trivial consistent finite set of height one identities. Then $\Meta(\Sigma)$ is \NP-complete. 
	\end{theorem}
	
	In particular, deciding whether $\fB$ has a Siggers polymorphism is \NP-complete~\cite{MetaChenLarose}.
	Note, however, that this theorem does not make any statement about deciding the existence of an idempotent Siggers polymorphism or the existence of a Maltsev polymorphism (since the idempotence of an operation cannot be expressed by a set of height one identities).

	If $\fB$ is additionally promised to be a core,
	then it is easy to see that there exists a Siggers polymorphism if and only if there exists an idempotent Siggers polymorphism.
	We do not know whether there exists a polynomial-time algorithm to test whether $\fB$ has an idempotent Siggers polymorphism, and the same question is open for Maltsev polymorphisms.
	
	Similarly to the case of the CSP, 
	we can relax the problem to a promise version as follows. Let $\Sigma_1$ and $\Sigma_2$ be 
	sets of $\sigma$-identities such that 
	$\Sigma_1 \models \Sigma_2$. 
	Then the \emph{promise metaproblem for $(\Sigma_1,\Sigma_2)$}, denoted by $\PMeta(\Sigma_1,\Sigma_2)$, is the following computational problem. 
	
	\comproblem{A finite structure $\fB$.}{
		`Yes' if $\Pol(\fB) \models \Sigma_1$, `No' if $\Pol(\fB) \not \models \Sigma_2$.}
    Moreover, if the set $\tau \subseteq \sigma$ of symbols that appear in $\Sigma_2$ 
    is finite, then the \emph{promise creation-metaproblem for $(\Sigma_1,\Sigma_2)$}, denoted by $\PCreaMeta(\Sigma_1,\Sigma_2)$, is the following computational problem. 
    \comproblem{A finite structure $\fB$.}{
		A map $f \colon \tau \to \Pol(\fB)$ witnessing that $\Pol(\fB) \models \Sigma_2$, 
        if $\Pol(\fB) \models \Sigma_1$.\newline
        If $\Pol(\fB) \not \models \Sigma_1$, the output may be anything.}
    Thus, the output consists of an interpretation of each function symbol in $\tau$.
    The following observation explains the name $\PCreaMeta$.
    
    \begin{observation}\label{obs:pcMetaToPMeta}
       If $\Sigma_2$ is finite and $\PCreaMeta(\Sigma_1,\Sigma_2)$ is in \P, then  $\PMeta(\Sigma_1,\Sigma_2)$ is in \P. 
    \end{observation} 
    
    This follows simply by applying the algorithm for $\PCreaMeta(\Sigma_1,\Sigma_2)$ and then testing whether the returned family of operations satisfies $\Sigma_2$ (this is indeed in \P as $\Sigma_2$ is finite and not part of the input here).
	An inspection of Chen and Larose's proof of Theorem~\ref{thm:h1} 
	shows that it proves the following stronger result.
	
	\begin{proposition}\label{prop:negative}
		Let $\Sigma_1$ and $\Sigma_2$ be finite sets of height one identities such that $\Sigma_1 \models \Sigma_2$, $\Sigma_1$ is consistent, and $\Sigma_2$ is non-trivial.
		Then $\PMeta(\Sigma_1,\Sigma_2)$ is $\NP$-hard. 
	\end{proposition}
    Note that we only state $\NP$-hardness in \cref{prop:negative}, because formally \NP only contains decision problems, not promise problems ($\Meta(\Sigma_1)$, in contrast, is of course in \NP).
    
	\begin{proof}[Proof sketch]
        The result is shown by a reduction from graph 3-colorability. For a given finite graph $G$, Chen and Larose~\cite[Lemma 6.6]{MetaChenLarose} construct a relational structure $\fB$ which has the property that $G$ is 3-colorable if and only if $\fB$ is not a core. Moreover, if $G$ is 3-colorable, then the construction is such that the polymorphisms of the core of $\fB$ satisfy all consistent idempotent sets of linear identities (see \cite[Lemma 6.6 (2)]{MetaChenLarose}).
        Thus, the polymorphism clone of the core of $\fB$ and therefore also $\Pol(\fB)$ itself satisfy all consistent sets of height one identities. 
		If $G$ is not 3-colorable, then $\Pol(\fB)$ has the property that
		it does not satisfy any non-trivial height one identity (a consequence of~\cite[Fact 4]{MetaChenLarose}).
		This shows the statement. 
	\end{proof}
	
	\section{Connection with the Uniform CSP}
	\label{sect:uniform} 
	The question whether certain tractability conditions for CSPs \emph{uniformize} has been asked and studied by Kolaitis and Vardi~\cite{KolaitisV00}. 
	Let $\Sigma$ be a set of identities. 
    The \emph{uniform CSP} for $\Sigma$ is defined as follows.
	
	\comproblem{A pair of finite structures $(\fA,\fB)$ of the same signature such that $\Pol(\fB) \models \Sigma$.}{
		`Yes' if there is a homomorphism from $\fA$ to $\fB$, `No' otherwise.}

	If $\fB$ does not satisfy $\Sigma$, then an algorithm for this problem might answer arbitrarily. 
    If $\Sigma$ is the set of identities that states the existence of a Maltsev (or Siggers, etc.) polymorphism, then the computational problem introduced above is called the \emph{uniform CSP for Maltsev constraints} (or for \emph{Siggers constraints}, etc.).
	It is an open problem whether there is a polynomial-time algorithm that solves the uniform CSP for Siggers constraints; this is stated as Problem~1 in Zhuk~\cite{Zhuk20} (there, the problem is stated using weak near unanimity polymorphisms, which, however, is equivalent as explained in Section~\ref{sect:prelims}).
    
   The following polynomial-time tractability result generalizes the second part of Theorem 4.7 in~\cite{MetaChenLarose}, which is the special case where $\Sigma_1 = \Sigma_2$.
    Be aware that $\Sigma_1$ is neither required to be finite nor linear.
    
	\begin{proposition}\label{prop:inP}
       Let $\Sigma_1$ be an idempotent set of identities such that the uniform $\CSP$ for $\Sigma_1$ is in \P, 
       and let $\Sigma_2$ be a finite set of linear identities such that $\Sigma_1 \models \Sigma_2$.
		Then $\PCreaMeta(\Sigma_1,\Sigma_2)$ and $\PMeta(\Sigma_1,\Sigma_2)$ are in \P.
	\end{proposition}

	\begin{proof}
        Let $\fB$ be a finite structure and let $\fC$ be the expansion of $\fB$ whose signature contains a symbol $R_=$ for the equality relation and a symbol $R_b$ for every element $b \in B$ which denotes the unary relation $\{b\}$.
        Note that if $\fB$ has polymorphisms satisfying $\Sigma_1$, then, as they are idempotent, these polymorphisms are also polymorphisms of $\fC$. 
        Therefore, $\fC \models \Sigma_1$. 
        Since $\Sigma_1 \models \Sigma_2$ and each polymorphism of $\fC$ must be idempotent, it follows that $\Sigma_1 \models \Sigma_2 \cup \Delta$ and $\fC \models \Sigma_2 \cup \Delta$, where $\Delta$ is a set of equations that express that all functions mentioned in $\Sigma_2$ are idempotent. 
        Conversely, if $\fB \not\models \Sigma_2$, then certainly $\fB \not \models \Sigma_2 \cup \Delta$ and $\fC \not\models \Sigma_2 \cup \Delta$.
        Therefore, from now on we can safely replace $\Sigma_2$ by $\Sigma_2 \cup \Delta$ and, hence, assume that $\Sigma_2$ is idempotent.

		We may suppose without loss of generality 
		that $\Sigma_2$ is the union of a set of height one identities and the identities of the form $f(x,\dots,x) \approx x$.

        The following construction is essentially the \emph{indicator structure} from~\cite{JeavonsClosure}.
		For each function symbol of arity $k$ that appears in $\Sigma_2$, we create a copy of $\fB^k$ (i.e., the structure with domain $B^k$ and the relations as defined in \cref{sect:prelims}).
		Let $\fI$ be the structure obtained from the disjoint union over all these copies as follows.
        For every $b \in B$, we add to $R_b^{\fI}$ all tuples
        of the form $(b,\dots,b)$ in the copies of $\fB$.
		We additionally add for every identity $f(x_{i_1},\dots,x_{i_k}) \approx g(x_{j_1},\dots,x_{j_{\ell}})$
        in $\Sigma_2$ 
        and every function 
        $\alpha \colon 
        \{x_{i_1},\dots,x_{i_k},x_{j_1},\dots,x_{j_{\ell}} \} \to B$
        the constraint
        $c = d$ where
        $c = (\alpha(x_{i_1}),\dots,\alpha(x_{i_k}))$
        belongs to the copy of $\fB^k$ created for $f$, and 
        $d = (\alpha(x_{j_1}),\dots,\alpha(x_{j_\ell}))$ belongs to the copy of $\fB^{\ell}$ created for $g$  
        (adding the constraint $c=d$ means we include $(c,d)$ to $R_=^\fI$).

		Observe that if there is a homomorphism $h$ from $\fI$ to $\fC$, then $\Pol(\fB) \models \Sigma_2$:
        each restriction of $h$ to a copy of $\fB^k$ is a polymorphism of $\fB$ of arity $k$, 
        and the additional equality constraints ensure that the respective operations satisfy the given identities $\Sigma_2$.
        Let $A$ denote the uniform polynomial-time algorithm for $\Sigma_1$. 
		Run $A$ to determine whether there is a homomorphism $\fI \to \fC$.
		If $A$ rejects, then return `No'. To see that this is the correct answer, suppose for contradiction that $\Pol(\fB) \models \Sigma_1$.
        Then the answer of the algorithm $A$ must be correct.
        Hence, there is no homomorphism from $\fI$ to $\fC$ and, therefore, $\Pol(\fB) \not \models \Sigma_2$, a contradiction to the assumption that $\Sigma_1$ implies $\Sigma_2$. 

        If $A$ accepts, we proceed like for the well-known self-reduction of SAT.
		If $A$ accepts, pick an element $x$ from $\fI$. For every $b \in B$, do the following:
		\begin{itemize}
			\item Add the constraint $R_b(x)$, i.e., add $x$ to $R_b^{\fI}$.
			\item Run $A$ on the resulting instance. 
			\item If $A$ accepts, then proceed in the same way with another element of $\fI$. 
		\end{itemize}
		If for some variable, $A$ rejects for all elements $b \in B$, then we return `No'. 
		Note that this answer must be correct: if $\Pol(\fB) \models \Sigma_1$, then it is impossible that we run into this case by the correctness of $A$ if the promise is fulfilled. If $\Pol(\fB) \not \models \Sigma_1$,
		then the answer `No' is correct. 
		
		Otherwise, eventually all elements $x$ of $\fI$ are contained in $R^{\fI}_b$ 
        for some $b \in B$.
      Thus, given the promise $\Pol(\fB) \models \Sigma_1$, the map $x \mapsto b$ is a solution to $\PCreaMeta(\Sigma_1,\Sigma_2)$.
        
       To solve $\PMeta(\Sigma_1,\Sigma_2)$, it remains to verify in polynomial time whether the map that sends $x$ to $b$ is a homomorphism from $\fI$ to $\fB$ (see \cref{obs:pcMetaToPMeta}).
       If no, return `No', otherwise return `Yes'. These answers are clearly correct. 
	\end{proof}

    A \emph{semiuniform algorithm for the $\CSP$ for $\Sigma$} (where $\Sigma$ is some set of identities) is an algorithm that takes as input a pair of finite structures $(\fA,\fB)$ of the same signature $\tau$ such that $\Pol(\fB) \models \Sigma$ together with a map $f \colon \tau \to \Pol(\fB)$ witnessing that $\Pol(\fB) \models \Sigma$ and outputs whether there is a homomorphism from $\fA$ to $\fB$.

    The following can be seen as a form of converse of 
    Proposition~\ref{prop:inP}: there, the existence of a uniform 
    algorithm for the CSP for $\Sigma_1$ implies that $\PCreaMeta(\Sigma_1,\Sigma_2)$ is in \P; here, the assumption that $\PCreaMeta(\Sigma_1,\Sigma_2)$ is in \P, together with a polynomial-time semiuniform algorithm for $\Sigma_2$, implies that the uniform algorithm for $\Sigma_1$ is in \P.

    \begin{proposition}\label{prop:pcMetaimpliesUniform}
    Let $\Sigma_1$ and $\Sigma_2$ be sets of $\sigma$-identities such that $\Sigma_1 \models \Sigma_2$ and the set $\tau \subseteq \sigma$ of symbols that appear in $\Sigma_2$ is finite. If $\PCreaMeta(\Sigma_1,\Sigma_2)$ is in \P and there is a polynomial-time semiuniform algorithm for the $\CSP$ for $\Sigma_2$, then the uniform $\CSP$ for $\Sigma_1$ is in \P.
    \end{proposition}
    \begin{proof}
        Given a pair $(\fA,\fB)$ of structures of the same signature such that $\Pol(\fB) \models \Sigma_1$. 
        By assumption, we can compute $f \colon \tau \to \Pol(\fB)$ witnessing that $\Pol(\fB) \models \Sigma_2$. We can then use the semiuniform algorithm for the $\CSP$ for $\Sigma_2$ to decide in polynomial time whether there is a homomorphism from $\fA$ to $\fB$.
    \end{proof}

	\section{Coset-Generating Polymorphisms}
	\label{sect:coset-gen}
    A \emph{coset-generating operation} on a set $A$ is an operation $m \colon A^3 \to A$ that can be written as $m(x,y,z) = x\circ y^{-1} \circ z$ for some operations $\circ$ and ${\cdot}^{-1}$ such that $(A,\circ, {\cdot}^{-1}, e)$ forms a group for some element $e \in A$. 
    In the following, we use juxtaposition instead of $\circ$.
	Let $G$ be a group. We write $H \leq G$ if $H$ is a subgroup of $G$. The following is well-known and the motivation why these operations are called \emph{coset-generating}; for the convenience of the reader we give a proof.
    
	\begin{lemma}
		\label{lem:folklore}
		Let $k \in {\mathbb N}$. Then a relation $R \sse G^k$ is preserved by the operation $m \colon G^3 \to G$ defined by $m(x,y,z) = x y^{-1} z$ if and only if $R = g H$ for some $g\in G^k$ and $H\leq G^k$ (\ie $R$ is a coset of a subgroup of $G^k$).
	\end{lemma}

	\begin{proof}
		Let $H$ be a subgroup of $G^k$. 
		Let $g \in G^k$ and $h_1,h_2,h_3 \in H$. 
		As usual, we may apply $m$ to elements of $G^k$ componentwise; 
		then \[m(gh_1,gh_2,g h_3) = gh_1 (g h_2)^{-1} gh_3 = gh_1 h_2^{-1} h_3 \in gH\]
		so $m$ indeed preserves all cosets of $H$.
		
		Conversely, suppose that $R \subseteq G^k$ is preserved by $m$. 
		Choose $y \in R$ arbitrarily. 
		We claim that $y^{-1} R$ is a subgroup $H$ of $G^k$. This will show that $R = yH$ is a coset of a subgroup of $G^k$. 
		Arbitrarily choose $a,b \in y^{-1} R$. 
		Then $x := ya \in R$ and $z := yb \in R$. 
		Hence, $m(x,y,z) = ya y^{-1} yb = y ab \in R$, so $ab \in y^{-1} R$ and 
		$y^{-1} R$ is closed under the group operation. 
		Moreover, 
		$m(y,ya,y) = y (ya)^{-1} y = y a^{-1} \in R$, 
		so $a^{-1} \in y^{-1} R$ and $y^{-1} R$ is also closed under taking inverses.
	\end{proof}

	A structure $\fB$ with a polymorphism which is a coset-generating operation 
    is called a \emph{coset structure} (with respect to $G$).
    CSPs for coset structures will be called \emph{coset CSPs}.
	As $xx^{-1}y= y = y x^{-1} x$, every coset-generating operation is a Maltsev operation. 
    Therefore, every coset structure has a Maltsev polymorphism, but the converse is not true (see, e.g., the Maltsev polymorphisms that appear in~\cite{BodirskyMoorheadConservativeMaltsev}),
    or the `No'-instances resulting from the reduction in \cref{thm:groupCSPNPcomplete}. 
    However, we have the following well-known fact. 

	\begin{lemma}\label{lem:heap}
		A relational structure $\fB$ is a coset structure if and only if it has a Maltsev polymorphism $m$ that additionally satisfies for all $x,y,u,v,w \in B$ the following associativity identity
		\begin{align}
			m(u,x,m(v,y,w)) & = m(m(u,x,v),y,w). 
			\label{eq:heap}
		\end{align}
	\end{lemma}

	\begin{proof}
		If $\fB$ is a coset structure and $f \colon (x,y,z) \mapsto x y^{-1} z$ is the coset-generating polymorphism of $\fB$, then $f$ is a Maltsev operation that clearly satisfies the associativity condition~\eqref{eq:heap}. 
		Conversely, if $m$ is a Maltsev operation that satisfies~\eqref{eq:heap}, then we can pick an arbitrary element $e \in B$, and define 
		$g \colon (x,y) \mapsto m(x,e,y)$. 
		It is straightforward to verify that $(B;g)$ is a group with identity element $e$ (the inverse of an element $x$ is given by $m(e,x,e)$), and that $m$ is the coset-generating polymorphism with respect to this group.
	\end{proof}

	Ternary Maltsev operations $m$ that satisfy~\eqref{eq:heap} are also called \emph{heaps}~\cite{heap} (the concept goes back to Pr\"ufer~\cite{Pruefer1924} and Baer~\cite{Baer1929} in the 1920s). 
	Thus, coset-generating polymorphisms (terminology used in~\cite{MetaChenLarose}) and heaps are the same concept; in the following, we use both terms interchangeably, but prefer `coset-generating' in the context of polymorphisms, and `heaps' in the abstract setting of clones.
    Note that sometimes also the identity $m(u, x, m(v, y, w)) = m(u, m(y, v, x), w)$ is required for heaps. However, this is implied by the other identities, which follows from our \cref{lem:heap} and can also be found in~\cite[Lemma 2.3]{Trusses}. 

    We write $\SigHeap$ for the identities that state the existence of a heap operation (i.e.,\ \eqref{eq:malt} together with \eqref{eq:heap}).
    Using \cref{prop:inP} and \cref{prop:pcMetaimpliesUniform}, we obtain the following corollary.

    \begin{corollary}\label{cor:cosetgenuniformPmeta}
        $\PCreaMeta(\SigHeap,\SigMaltsev)$ is in \P if and only if there is a uniform polynomial-time algorithm for coset CSPs.
	\end{corollary}

    In the light of \cref{prop:negative} and \cite{MetaChenLarose}, it might be interesting to know whether $\SigHeap$ can be enforced by some consistent set of linear identities.  
    As we see next, this is not the case.

	\begin{proposition}\label{obs:notMaltsev}
		Let $\Sigma$ be a set of linear identities with $\Sigma \not\models x\approx y$.
		Then there is a structure $\fB$ without a coset-generating polymorphism such that $\Pol(\fB) \models \Sigma$.
		In particular, being a coset structure cannot be defined by linear identities.
	\end{proposition}
	
	\begin{proof}
		Let $\fA$ be a structure with domain $A$ such that $\Pol(\fA) \models \Sigma$ and $n := \abs{A} \geq 2$  
		(such an $\fA$ exists because $\Sigma \not\models x\approx y$).
        Define a structure $\fB$ with domain $B = A \cup \{b\}$ and with relations $R_i^{\fB} = R_i^{\fA}$; furthermore, add the unary relation $A \sse B$.
        Now, $\Pol(\fB)$ clearly consists of all extensions of polymorphisms of $\fA$ (meaning that the restriction of each $f \in \Pol(\fB)$ of arity $k$ to $A^k$ is in $\Pol(\fA)$). 
        Therefore, by \cref{lem:extend_domain}, $\Pol(\fB) \models \Sigma$.
		
		It remains to see that $\fB$ does not contain a coset-generating operation. Suppose otherwise.
		Then the  corresponding group  would have order $n+1$.
        The size of each relation $R \sse B^r$ of $\fB$ needs to divide $(n+1)^r$: 
        the relation $R$ must be a coset of a subgroup of a group with domain $B^r$ by Lemma~\ref{lem:folklore}, and hence
        $|R|$ must divide $|B|^r$ by Lagrange's theorem.
		However, the unary relation $A$ contains $n$ elements, which is coprime to $(n+1)^r$~-- thus, $\Pol(\fB)$ cannot contain a coset-generating operation. 
	\end{proof}

\section{NP-Hardness of Detecting Heap Polymorphisms}\label{sect:coset-gen-proof}
	
	The following is our main result and answers the question of Chen and Larose~\cite{MetaChenLarose} mentioned in the introduction. 
	
	\begin{theorem}[see \cref{thm:groupCSPNPcompleteIntro}]\label{thm:groupCSPNPcomplete}
		The following problem is \NP-complete:
		\compproblem{A finite structure $\fB$}{Is $\fB$ a coset structure?}
	\end{theorem}

    \begin{corollary}\label{cor:groupCSPNPcomplete}
        The creation-metaproblem for coset-generating polymorphisms is not in \P unless $\P = \NP$.
    \end{corollary}

    \begin{proof}
        A polynomial-time algorithm for the creation-metaproblem can be used to solve the metaproblem in \P by running the algorithm and then checking whether the output is indeed a coset-generating polymorphism. 
        The latter can be done in polynomial time using the identities from \cref{lem:heap} (see \cref{obs:pcMetaToPMeta}). The statement therefore follows from Theorem~\ref{thm:groupCSPNPcomplete}.
    \end{proof}

    \begin{remark} 
        By \cref{cor:groupCSPNPcomplete}, if there exists a uniform polynomial-time algorithm for coset CSPs, it cannot use the approach of first computing the coset-generating polymorphism and then apply the semiuniform algorithm from \cite{BulatovDalmau}.
    
        Note that there still might be a uniform $\P$-time algorithm for coset CSPs.   
        Indeed, \cite[Corollary 4.9]{MetaChenLarose} does not apply to this situation, because the property to be a coset structure cannot be defined by linear identities, as we have seen in Proposition~\ref{obs:notMaltsev}.
    \end{remark}

	Our proof of \cref{thm:groupCSPNPcomplete} is based on a reduction from the following graph problem, which has been shown to be \NP-complete in \cite[Theorem 4.9]{CowenGJ97} (where the problem is called $(2,1)$ graph coloring; see also \cite{LimaRSS18}).

	\begin{lemma}\label{lem:graphNPcomplete}
		The following problem is \NP-complete:
		\compproblem{A finite undirected simple graph $G = (V,E)$}{Can $E$ be written as the disjoint union of a matching and the edge set of a bipartite graph?}
	\end{lemma}
    In other words, the question in the computational problem is whether one can remove a matching so that the remaining graph becomes bipartite.
    For convenience of the reader, we give a proof in the appendix.

    To reduce the computational problem from Lemma~\ref{lem:graphNPcomplete} to the question whether a given structure has a coset-generating polymorphism, we need some basic notions from group theory.
    
    \subparagraph*{Group basics.}
    For $n \in {\mathbb N} \setminus \{0\}$, we write $C_n$ for the cyclic group of order $n$. If $G$ and $H$ are groups, then we write $G \times H$ for the direct product of $G$ and $H$.
    Given groups $G,H$ and a homomorphism $\phi$ from $H$ to the automorphism group $\Aut(G)$ of $G$ (or equivalently, an action of $H$ on $G$ via automorphisms) the semi-direct product $G \rtimes_\phi H$ is the group with elements $\{ (g,h)\mid g \in G, h\in H\}$ and multiplication defined by
    \[(g_1, h_1) \cdot (g_2, h_2) = (g_1 \,\phi(h_1)(g_2),\, h_1 h_2).\]
    In this case, $G$ (identified with $G \times \{1\}$) is a normal subgroup of $G \rtimes_\phi H$. 
    Then the map $\pi$ that sends each element $f$ of $G \rtimes_\phi H$ to the coset $f G $ of $G$ is a group homomorphism $G \rtimes_\phi H \to (G \rtimes_\phi H) / G \cong H$. The map $\iota \colon h \mapsto (1,h)$ is an embedding of $H$ into $G \rtimes_\phi H$ such that the composition $H \overset{\iota}{\to} G\rtimes_\phi H \overset{\pi}{\to} H$ is the identity (note that $\iota$ is not unique with this property). 
    For simplicity, we simply write $G\rtimes H$ instead of $G\rtimes_\phi H$ and specify $\phi$ only verbally.
    For instance, if $G = C_p$ and $H = C_{2k}$, then \emph{$G \rtimes_\phi H$ with the action by inversion} means that 
    $\phi$ is the homomorphism from $C_{2k}$ to $\Aut(C_p)$ which sends the generator $a$ of $C_{2k}$ to the map $x \mapsto x^{-1} \in \Aut(C_p)$.
    
    Another example that is relevant later is $G = C_p$ for $p \equiv 1 \mod 4$ and again $H = C_4$, where \emph{$G \rtimes_\phi H$ with a faithful action} means that $\phi$ is an injective homomorphism from $C_4$ to $\Aut(C_p)$. This means that $\phi$ maps the generator $a$ of $C_4$ to the map $x \mapsto x^{k} \in \Aut(C_p)$ for $k = \pm (p-1)/4$. Note that the two different choices of $k$ give isomorphic groups.

	In order to prove our main theorem, we focus on dihedral groups. The dihedral group $D_{2n}$ of order $2n$ for $n \geq 1$, is given by the following group presentation:
	\[D_{2n} = \Gen{d,s}{d^{n} = s^2 = 1, ds = sd^{-1}}.\]
	It is well-known that $D_{2n}$ is isomorphic to the symmetry group of a regular $n$-gon and, thus, can be understood geometrically (see, e.g.,~\cite[Theorem I.6.13]{Hungerford}). 
	Each element of $D_{2n}$ has a unique normal form $s^{k}d^{\ell}$ with $k \in \{0,1\}$ and $\ell \in \{0, \dots, n-1\}$.
    If $k=1$, then $s^{k}d^{\ell}$ corresponds to a reflection, otherwise to a rotation.
    From this normal form, we can see that $D_{2n}$ is isomorphic to the semidirect product $C_n \rtimes C_2$ with the action by inversion.

	\begin{lemma}\label{lem:cosetsDfourp}
		Let $m \in {\mathbb N}$. For the dihedral group $D_{4m} = \Gen{d,s}{d^{2m} = s^2 = 1, ds = sd^{-1}}$ the cosets of subgroups of order two are precisely as follows: 
		\[\set{ \{d^k,sd^\ell\} }{k,\ell \in \{0,\dots, 2m-1\}} \cup \set{\{a,ad^m\}\vphantom{d^\ell}}{a \in D_{4m}}.\]
	\end{lemma}
	
	\begin{proof}
		The lemma follows from the fact that the subgroups of $D_{4m}$ of order two are the subgroups generated by $sd^k$ for $k \in \{0, \dots, 2m-1\}$ (reflections at different axes) and the subgroup generated by $d^m$ (rotation by 180 degrees).
	\end{proof}
	
	In particular, the graph $\Gamma = (V,E)$ defined by $V := D_{4m}$ and 
	\[E := \big\{ \{u,v\} \mid \{u,v\} \text{ is a coset of a subgroup of } D_{4m} \big\}\] can be written as the union of the complete bipartite graph $K_{2m,2m}$ and a perfect matching on both sides. 

    \begin{remark}\label{rem:odd_dihedral_graph}
    For an odd integer $r$, the corresponding graph for the dihedral group of order $2r$ is the complete bipartite graph $K_{r,r}$ (as the dihedral group then does not contain a rotation by 180 degrees).
    \end{remark} 
	
	\begin{proof}[Proof of \cref{thm:groupCSPNPcomplete}]
		We reduce the problem from \cref{lem:graphNPcomplete} to our problem.
		
		Let $\Gamma = (V,E)$ be a given graph with $n = \abs{V}$.
		We may assume without loss of generality that $\Gamma$ has a vertex of degree at least four (otherwise, we can add such a vertex to $\Gamma$).
		First, compute a prime $p \geq 5$ with $n \leq 2p$ (this is clearly possible by checking whether any of the integers between $n/2$ and $n$ is prime and by the Bertrand--Chebyshev theorem the existence of such a prime is guaranteed).
		
		The domain of the structure $\fB$ will be $B := \{1, \dots, 4p\}$; we consider $V$ as a subset of $B$. For the signature $\tau$ of $\fB$, it will be convenient to use the finite set $E$ as the index set for the  entries of the tuple $\tau$, rather than $\{1,\dots,|E|\}$; define 
		$\tau \in \N^{E}$ as $\tau := (1, \dots, 1)$. 
		For $\{u,v\} \in E$, we have $R^{\fB}_{\{u,v\}} := \{u,v\}$ (thus, in particular, $\fB$ has only unary relations, each of which contains exactly two elements; moreover, the elements of $B \setminus V$ are not contained in any relation).
		Clearly, given $\Gamma$, the structure $\fB$ can be computed in logarithmic space.
		We claim that $\fB$ has a coset-generating polymorphism if and only if the edges of $\Gamma$ can be partitioned into a bipartite graph and a matching.

		First, assume that $\fB$ has a coset-generating polymorphism.
		Then, by \cref{lem:folklore}, its relations are cosets of subgroups of some group $G$ of order $4p$.
        It is well known\footnote{See e.g.\ the groupprops wiki \url{https://groupprops.subwiki.org/wiki/Classification_of_groups_of_order_four_times_a_prime_congruent_to_1_modulo_4} (accessed on April 27, 2026)} 
         that, since $p \geq 5$ is prime, there are up to isomorphism, at most five groups of order $4p$, namely 
		\begin{enumerate}
			\item $C_2 \times C_2 \times C_p$, 
			\item $C_4 \times C_p \cong C_{4p}$, 
			\item $C_2 \times D_{2p} \cong D_{4p}$,  
			\item $C_p \rtimes C_4$ with the action by inversion (dicyclic group), and 
			\item possibly $C_p \rtimes C_4$ with a faithful action (if $p\equiv 1 \mod 4$).
		\end{enumerate} 
        Clearly, each of these groups has order $4p$; the fact that 
        every group of order $4p$ is isomorphic to one of the groups in this list can be shown by elementary group theory (essentially the Sylow theorems; see, e.g.~\cite[Chapter 1, Section 6]{Lang}). 
        In particular, one first shows that in every such group the (unique) $p$-Sylow subgroup is normal (see~\cite[Exercise 28]{Lang}).
               
		In the abelian cases (Item 1 and 2) and in the case 
		$G \cong C_p \rtimes C_4$ with the action by inversion (Item 4), 
		there are at most three subgroups of order two; therefore, each element can be contained in at most three relations.
		Hence, these cases are ruled out, because we required that $\Gamma$ has a vertex of degree four.
		
		In the case $C_p \rtimes C_4$ with the faithful action (Item 5), the subgroups of order two are all contained in the unique subgroup isomorphic to $D_{2p}$ (this is because under the projection $C_p \rtimes C_4 \to C_4$ an element of order two must map to the element of order two in $C_4$; thus, all order-two elements are contained in the unique subgroup of $C_p \rtimes C_4$ that is isomorphic to $C_p\rtimes C_2 $).
        Therefore, according to \cref{rem:odd_dihedral_graph}, $\Gamma$ would have to be a subgraph of the disjoint union of two copies of $K_{p,p}$ (one for each coset of $D_{2p}$)~-- in particular, $\Gamma$ is bipartite.
			
		Therefore, the only possible group is $D_{4p}$.
		In this case, by \cref{lem:cosetsDfourp}, we know how the cosets of subgroups of order two look like. 
		In particular, $E$ can be written as the union of a bipartite graph and a matching.

		Conversely,  assume that $\Gamma$ can be written as the union of a bipartite graph and a matching.
		Write $V = C \cup D$ where $C$ and $D$ are the two sides of the bipartite graph~-- meaning that the edges within $C$ as well as the edges within $D$ form a matching. 
		Then we can choose an injection $f \colon V \to D_{4p}$ such that 
		\begin{itemize}
			\item $f(C) \sse \gen{d}$, 
			\item $f(D)\sse s\gen{d}$, and 
			\item for every edge $\{u,v\}$ with both $u,v \in C$ or both $u,v \in D$ we have $f(u) = f(v) d^p$.
		\end{itemize}
		This is possible because $\abs{C},\abs{D} \leq \abs{\gen{d}}$ and the edges within $C$ and $D$ form a matching.    
		Hence, by \cref{lem:cosetsDfourp}, all relations of $\fB$ are cosets of subgroups of $D_{4p}$; thus, $\fB$ has a coset-generating polymorphism.
	\end{proof}
    
    Observe that the structure $\fB$ in this proof \emph{always} has a Maltsev polymorphism: indeed, any map $m \colon B^3 \to B$ satisfying the Maltsev identities leaves the defined relations invariant as they are all unary and contain exactly two elements.
	
	\section{Abelian Heap Polymorphisms}

	In the context of the metaproblem for coset-generating  polymorphisms, Chen and Larose~\cite{MetaChenLarose} also state ``One can also ask these questions for restricted classes of groups.''
	In particular, the complexity of the following computational problem remains open. 
	\compproblem{A structure $\fB$}{Is $\fB$ a coset structure for some abelian group?}   
    
	It is easy to see that a structure $\fB$ is a coset structure for some abelian group if and only if it is preserved by a heap operation $m$ which is \emph{abelian}\footnote{It is well known and easy to show that abelian heap operations are abelian in the sense of universal algebra (we refrain from giving the definition, since we only need it in this footnote), and every heap operation which is abelian in the sense of universal algebra is an abelian heap operation; this justifies the terminology.}, i.e., satisfies 
	\[m(x,y,z) \approx m(z,y,x).\]
    We write $\SigAb$ for the identities that state the existence of an abelian heap operation.

	\begin{corollary}\label{cor:abheap}
        $\PMeta(\SigAb, \SigMaltsev)$ is in \P. 
	\end{corollary}
	\begin{proof}
		If $\fB$ has an abelian heap polymorphism, then by  Lemma~\ref{lem:heap} it is preserved by $(x,y,z) \mapsto x-y+z$ with respect to some abelian group $G$. Hence, for every $n \geq 2$ it is also preserved by $(x_1, y_1, x_2 ,y_2,\dots, x_n) \mapsto x_1-y_1+x_2-y_2+\cdots+x_n$, which is 
		an example of a so-called \emph{alternating polymorphism}.
		It is known that for such $\fB$, the AIP algorithm mentioned in the introduction is a uniform polynomial-time algorithm~\cite[Theorem 7.19]{BartoBKO21}.
        Hence, the statement follows from Proposition~\ref{prop:inP}. 
	\end{proof}
	
	Note that $\PMeta(\SigAb,\SigMaltsev)$ is in \P despite the fact that the polynomial-time tractability of $\Meta(\SigAb)$ and of $\Meta(\SigMaltsev)$ is open. 

\section{Conclusion and Open Problems}\label{sect:open}
    
    In this work we have introduced the promise version of the metaproblem in constraint satisfaction, and showed that several known results about the complexity of the metaproblem and connection to the existence of uniform polynomial-time CSP algorithms extend naturally to this larger setting. 
    We identified a pair of 
    polymorphism conditions $(\Sigma_1,\Sigma_2)$ whose promise metaproblem is in \P, but 
    where the metaproblem for $\Sigma_1$ and the metaproblem for $\Sigma_2$ are not known to be in \P. 
    We then proved that the problem of deciding whether a given finite structure has a coset-generating (i.e., heap) polymorphism is \NP-complete, solving an open problem of Chen and Larose~\cite{MetaChenLarose}.

Before the present work, there were several plausible scenarios how the question for the uniform complexity for coset CSPs might be settled: 
\begin{enumerate}[(A)]
	\item there could be a polynomial-time algorithm that given a finite structure promised to have a coset-generating polymorphism  
    actually computes a corresponding group.
    If this is the case, then one obtains a uniform polynomial-time algorithm for such CSPs by using Bulatov and Dalmau's semiuniform algorithm for CSPs with Maltsev constraints~\cite{Maltsev};
	\item the group structure could be difficult to compute, yet one could compute in polynomial time a Maltsev polymorphism (and, again, the algorithm from \cite{Maltsev} can be used to decide the CSP); 
	\item there could be no polynomial-time algorithm to compute a Maltsev polymorphism, but still coset CSPs could be uniformly decidable in polynomial time by another algorithm;
	\item the uniform CSP for coset structures could be hard~-- here, besides being \NP-complete, there are other hardness results imaginable: for instance, it might be hard for graph isomorphism or hard for group isomorphism.
\end{enumerate}
In this work we ruled out possibility (A) by \cref{cor:groupCSPNPcomplete} (unless $\P = \NP$) and possibility (C) by \cref{cor:cosetgenuniformPmeta}.
In particular, if there is a uniform polynomial-time algorithm for coset CSPs, then a Maltsev polymorphism can be computed in \P, but no coset-generating polymorphism can be computed in general (unless $\P = \NP$).

    We ask the following questions that are related to this work. 
    \begin{itemize}
       \item Is there a uniform polynomial-time algorithm for coset CSPs? 
        \item What is the answer to this question restricted to domains of size $4p$, where $p$ is a prime? (This question is motivated 
        by our proof of \cref{thm:groupCSPNPcompleteIntro}.)
        \item What is the complexity of deciding whether a given structure has an abelian heap polymorphism?
        \item What is the complexity of 
        deciding whether a given structure 
        has an abelian Maltsev polymorphism?
    \end{itemize}
The complexity of the metaproblem and the uniform CSP for many natural finite sets of linear identities remains open. We wonder whether the class of (promise) metaproblems for fixed sets of identities is perhaps a candidate for a class that exhibits a complexity non-dichotomy.

\bibliography{global}

\appendix

\section{Proof of \cref{lem:graphNPcomplete}}

    For convenience of the reader, we restate and give a proof of \cref{lem:graphNPcomplete}. 
	\begin{lemma}[\cref{lem:graphNPcomplete} restated]\label{lem:appendix}
		The following problem is \NP-complete:
		\compproblem{A finite undirected simple graph $G = (V,E)$}{Can $E$ be written as the disjoint union of a matching and the edge set of a bipartite graph?}
	\end{lemma}

    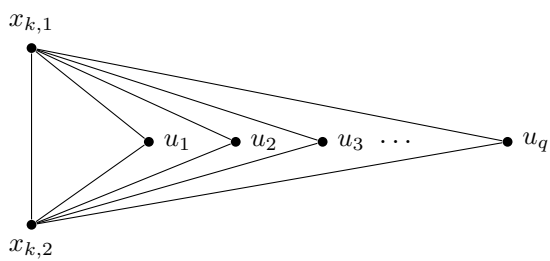
\begin{figure}
        \centering
        \begin{tikzpicture}[
            node distance=1.4cm,
            circ/.style={circle,fill,inner sep=1.3pt}
            ]
            
            % Base nodes
            \node[circ,label=below:$x_{k,2}$] (y) {};
            \node[circ,above=2.2cm of y,label=above:$x_{k,1}$] (x) {};
            
            % Middle row (between x and y)
            \node[circ,right=1.4cm of y, yshift=1.1cm, label=right:$u_1$] (m1) {};
            \node[circ,right=1.0cm of m1,              label=right:$u_2$] (m2) {};
            \node[circ,right=1.0cm of m2,              label=right:$u_3$] (m3) {};
            
            % Three dots
            \node[right=0.55cm of m3] (dots) {\large $\dots$};
            
            % u_q, moved further to the right
            \node[circ,right=1.0cm of dots, label=right:$u_q$] (m4) {};
            
            % Connections
            \draw (y) -- (m1) -- (x);
            \draw (y) -- (m2) -- (x);
            \draw (y) -- (m3) -- (x);
            \draw (y) -- (m4) -- (x);
            
            % Direct edge
            \draw (x) -- (y);
            
        \end{tikzpicture}
        \caption{For each variable $X_k$ ($k \in \{1, \dots, n\}$) we have one of these subgraphs in the proof of Lemma~\ref{lem:graphNPcomplete}.
        Here the vertices $u_1, \dots, u_q$ correspond to those vertices $v_{i,j}$ where the literal $L_{i,j}$ is the variable $X_k$. Note that $q \geq 3$.
        }
        \label{fig:gadgetnew}
    \end{figure}

    \begin{proof}
		Let us show a reduction from positive not-all-equal 3SAT~\cite{GareyJohnson,Schaefer}. 
		Consider an instance \[\Phi = \big \{\vphantom{d^\ell}\{L_{i,1},L_{i,2},L_{i,3}\} \mid i \in \{1, \dots, m\}\big \}\] of this problem with $m$ clauses and $n$ variables $X_1, \dots, X_n$ (\ie $L_{i,j} \in \{X_1, \dots, X_n\}$ for all $i \in \{1,\dots,m\}$ and $j \in \{1,2,3\}$). One may assume without loss of generality that each variable appears in at least three clauses (by simply duplicating clauses). We construct a graph $(V,E)$ with vertex set 
        \[V = \big \{v_{i,j} \mid i \in \{1,\dots, m\}, j \in \{1,2,3\}\big \} \cup \big\{x_{k,\ell} \mid k\in \{1, \dots, n\}, \ell \in \{1,2\} \big\}\] 
        (i.e.\ we have one vertex for each literal and two additional vertices for every variable) and edge set $E$ such that 
        \begin{itemize}
            \item $\{v_{i,1},v_{i,2},v_{i,3}\}$ becomes a clique for every $i \in \{1,\dots,m\}$,
            \item there is an edge $\{x_{k,1}, x_{k,2}\}$ for every $k \in \{1,\dots,n\}$,
            \item whenever $L_{i,j} = X_k$, there are the edges $\{v_{i,j},x_{k,1}\}$ and $\{v_{i,j},x_{k,2}\}$.
        \end{itemize}
        Thus, \begin{align*}
            E &= \big \{\{v_{i,j},v_{i,k}\} \mid i \in \{1,\dots, m\}, j\neq k \in \{1,2,3\} \big \}\\
            &\cup \big\{\{x_{k,1}, x_{k,2}\} \mid k \in\{1, \dots, n\} \big\} \\
            &\cup \big\{\{v_{i,j},x_{k,\ell}\} \mid i \in \{1,\dots, m\}, j \in \{1,2,3\},  k \in\{1, \dots, n\} ,\ell \in \{1,2\} \text{ with } L_{i,j} = X_k \big\}.
        \end{align*}
        The edges of the second and third type associated to a single variable $X_k$ are illustrated in Figure~\ref{fig:gadgetnew}~-- the edges of the first type connect these respective gadgets associated to different variables.
        Clearly, this graph can be computed in polynomial time from the given instance $\Phi$. 

        We claim that $\Phi$ has a solution if and only if the given graph can be written as a disjoint union of a matching and the edge set of a bipartite graph.
		In order to see this, first note that each of the cliques $\{v_{i,1},v_{i,2},v_{i,3}\}$ can be written as a union of a bipartite graph and a matching (e.g., the edges $\{v_{i,j},v_{i,2}\}$
		for $j \in \{1,3\}$ form a bipartite graph, and the edges $\{v_{i,1},v_{i,3}\}$ form a matching).
		Moreover, any such edge decomposition must have one vertex of each triangle on one side of the bipartite graph and the other two vertices on the other side (and there is precisely one matching edge).

        Now assume that $\Phi$ has a satisfying assignment $\sigma \colon \{X_1, \dots, X_n\} \to \{0,1\}$ (\ie for every $i \in \{1,\dots,m\}$ we have $1\leq \sigma(L_{i,1}) + \sigma(L_{i,2}) + \sigma(L_{i,3}) \leq 2$).
        Setting
        \[A = \{v_{i,j}\mid \sigma(L_{i,j})=1\} \cup \{x_{k,\ell} \mid \sigma(X_k) = 0\} \]
        and $B = V\setminus A$ gives us a bipartite decomposition of the graph after removing the matching consisting of all the edges $\{x_{k,1}, x_{k,2}\}$ for $k \in \{1, \dots, n\}$ and $\{v_{i,j} v_{i,k}\}$ whenever $ \sigma(L_{i,j}) = \sigma(L_{i,k}) $ with $i \in \{1, \dots, m\}$ and $j\neq k \in \{1,2,3\}$.

        To see the other direction of the correctness proof, observe that in any decomposition of $E$ into a matching and bipartite graph, all the edges $\{x_{k,1}, x_{k,2}\}$ must be part of the matching. 
        This is because by our assumption $x_{k,1}$ and $x_{k,2}$ have at least three common neighbours, see \cref{fig:gadgetnew}.
        (Each of the respective triangles containing the edge $\{x_{k,1}, x_{k,2}\}$ must contain one matching edge.
        Since there are at least three of these triangles, the only possibility is the edge $\{x_{k,1}, x_{k,2}\}$.)

        Now, let $V = A \cup B$ be a bipartite decomposition of $V$ after removing a suitable matching.
        From the previous observation, we conclude that if $L_{i,j} = L_{\mu, \nu}$, then $v_{i,j}$ and $v_{\mu, \nu}$ are either both in $A$ or both in $B$, because they are connected by a path of length two that cannot contain a matching edge.
        Therefore, the assignment $\sigma \colon \{X_1, \dots, X_n\} \to \{0,1\}$ with $\sigma(X_k) = 1$ if and only if $x_{k,1} \in A$ is a satisfying assignment.
	\end{proof}

\end{document}